\documentclass[12pt,a4paper]{amsart}

\usepackage{amsmath,amsfonts,amssymb,mathtools,extarrows}

\usepackage[hmargin=2cm,vmargin=2cm]{geometry}

\usepackage[hypertexnames=false,hyperfootnotes=false,colorlinks=true,linkcolor=blue,%
citecolor=purple,filecolor=magenta,urlcolor=cyan,unicode,linktocpage=true,pagebackref=false]{hyperref}
\usepackage{nameref,zref-xr}     

\usepackage{enumerate}      

\usepackage{scalerel}         
\usepackage{stmaryrd}

\usepackage[all]{xy}

\usepackage{xcolor}

\newtheorem{dummy}{}[section]

\newtheorem{prop}[dummy]{Proposition}

\newtheorem{lemma}[dummy]{Lemma}
\newtheorem{cor}[dummy]{Corollary}

\newtheorem{conj}{Conjecture}
\newtheorem{theorem}{Theorem}
\theoremstyle{definition}

\theoremstyle{remark}

\newcommand{\M}{\ensuremath{\overline{\mathcal{M}}}}

\renewcommand{\d}{\ensuremath{\partial}}

\newcommand{\eps}{\varepsilon}
\newcommand{\<}{\left<}
\renewcommand{\>}{\right>}
\newcommand{\mbZ}{\mathbb{Z}}
\newcommand{\mbC}{\mathbb{C}}
\DeclareMathOperator{\res}{res}
\newcommand{\Coef}{\mathrm{Coef}}

\newcommand{\hL}{\widehat{L}}
\newcommand{\oM}{\overline{\mathcal{M}}}
\newcommand{\mbQ}{\mathbb{Q}}
\newcommand{\mcF}{\mathcal{F}}

\numberwithin{equation}{section}

\begin{document}

\title[Open $r$-spin theory III]{Open $r$-spin theory III: a prediction for higher genus}

\author{Alexandr Buryak}
\address{A. Buryak:\newline Faculty of Mathematics, National Research University Higher School of Economics, 6 Usacheva str., Moscow, 119048, Russian Federation;\smallskip\newline 
Center for Advanced Studies, Skolkovo Institute of Science and Technology, 1 Nobel str., Moscow, 143026, Russian Federation}
\email{aburyak@hse.ru}

\author{Emily Clader}
\address{E.~Clader:\newline San Francisco State University, San Francisco, CA 94132-1722, USA}
\email{eclader@sfsu.edu}

\author{Ran J. Tessler}
\address{R.~J.~Tessler:\newline Incumbent of the Lilian and George Lyttle Career Development Chair, Department of Mathematics, Weizmann Institute of Science, POB 26, Rehovot 7610001, Israel}
\email{ran.tessler@weizmann.ac.il}

\begin{abstract}
In our previous two papers, we constructed an $r$-spin theory in genus zero for Riemann surfaces with boundary and fully determined the corresponding intersection numbers, providing an analogue of Witten's $r$-spin conjecture in genus zero in the open setting. In particular, we proved that the generating series of open $r$-spin intersection numbers is determined by the genus-zero part of a special solution of a certain extension of the Gelfand--Dickey hierarchy, and we conjectured that the whole solution controls the open $r$-spin intersection numbers in all genera, which do not yet have a geometric definition. In this paper, we provide geometric and algebraic evidence for the correctness of this conjecture.
\end{abstract}

\date{\today}

\maketitle

\section{Introduction}

One of the most important results in the study of the intersection theory on the moduli spaces of stable curves $\oM_{g,n}$ is Witten's conjecture \cite{Witten2DGravity}, proved by Kontsevich \cite{Kontsevich}, saying that the generating series of intersection numbers
\begin{gather*}
\mcF^c(t_0,t_1,\ldots,\eps)=\sum_{g\ge 0}\eps^{2g-2}\mcF^c_g(t_0,t_1,\ldots):=\hspace{-0.15cm}\sum_{\substack{g \geq 0,\,n\geq 1\\2g-2+n>0}} \sum_{d_1,\ldots,d_n\geq 0} \frac{\eps^{2g-2}}{n!}\left(\int_{\M_{g,n}}\hspace{-0.15cm}\psi_1^{d_1} \cdots \psi_n^{d_n} \right)t_{d_1} \cdots t_{d_n}
\end{gather*}
is the logarithm of a tau-function of the KdV hierarchy. Here, $\psi_i \in H^2(\M_{g,n},\mbQ)$ is the first Chern class of the cotangent line bundle corresponding to the $i$-th marked point, $t_0,t_1,\ldots$ and~$\eps$ are formal variables, and the superscript ``$c$," which stands for ``closed," is to contrast with the open theory discussed below.  

\medskip

Witten also proposed a much more general conjecture, the so-called $r$-spin Witten conjecture~\cite{Witten93}, which considers the moduli space of stable curves with $r$-spin structure. On a smooth marked curve $(C;z_1, \ldots, z_n)$, an \emph{$r$-spin structure} is a line bundle~$S$ together with an isomorphism $S^{\otimes r} \cong \omega_{C}\left(-\sum_{i=1}^n a_i[z_i]\right)$, where $a_i \in \{0,1,\ldots, r-1\}$ and $\omega_C$ denotes the canonical bundle. There is a natural compactification $\M_{g,(a_1, \ldots, a_n)}^{1/r}$ of the moduli space of smooth curves with $r$-spin structure, and this space admits a virtual fundamental class $c_W\in H^*(\M_{g,(a_1, \ldots, a_n)}^{1/r},\mbQ)$ known as \emph{Witten's class}. In genus zero, $c_W$ is the Euler class of the derived pushforward $(R^1\pi_*\mathcal{S})^{\vee}$, where $\pi\colon\mathcal{C} \rightarrow \M_{0,(a_1, \ldots, a_n)}^{1/r}$ is the universal curve and~$\mathcal{S}$ is the universal $r$-spin structure. In higher genus, the sheaf $R^1\pi_*\mathcal{S}$ may not be a vector bundle, and the definition of Witten's class is much more intricate; see \cite{PV,ChiodoWitten,Moc06,FJR,CLL} for various constructions.

\medskip

Witten's $r$-spin conjecture, proved by Faber--Shadrin--Zvonkine~\cite{FSZ10}, states that if $t^a_d$ are formal variables indexed by $0\le a\le r-1$ and $d\ge 0$, then the generating series
\begin{align*}
\mcF^{\frac{1}{r},c}(t^*_*,\eps)=&\sum_{g\ge 0}\eps^{2g-2}\mcF^{\frac{1}{r},c}_g(t^*_*)\\
:=&\sum_{\substack{g \geq 0,\, n \geq 1\\2g-2+n>0}} \sum_{\substack{0 \leq a_1, \ldots, a_n \leq r-1\\ d_1, \ldots, d_n \geq 0}} \frac{\eps^{2g-2}}{n!}r^{1-g}\left(\int_{\M^{1/r}_{g,(a_1, \ldots, a_n)}} c_W\cdot \psi_1^{d_1} \cdots \psi_n^{d_n}\right) t^{a_1}_{d_1} \cdots t^{a_n}_{d_n}
\end{align*}
is, after a simple change of variables, the logarithm of a tau-function of the $r$-th Gelfand--Dickey hierarchy.

\medskip

While curves with $r$-spin structure generalize the study of $\M_{g,n}$ in one direction, a different direction was taken up by Pandharipande, Solomon, and the third author in \cite{PST14}, in which the study of intersection theory on the moduli space $\oM_{g,k,l}$ of Riemann surfaces with boundary was initiated.  Here, the genus $g$ of a Riemann surface with boundary $(C,\d C)$ is defined as the genus of the closed surface obtained by gluing two copies of $C$ along the boundary $\d C$, and the numbers $k$ and $l$ are the numbers of boundary and internal marked points, respectively. In~\cite{PST14}, intersection numbers on $\M_{0,k,l}$, called \emph{open intersection numbers} and denoted by
$$
\<\tau_{d_1}\cdots\tau_{d_l}\sigma^k\>^o_0,\quad d_i\ge 0,
$$
were defined and explicitly computed. Moreover, in~\cite{PST14} the authors proposed a conjectural description of open intersection numbers in all genera, which can be viewed as an open analogue of Witten's conjecture on $\oM_{g,n}$. A geometric construction of open intersection numbers in all genera, denoted by $\<\tau_{d_1}\cdots\tau_{d_l}\sigma^k\>^o_g$, was given by Solomon and the third author, hence the generating function 
$$
\mcF^o(t_0,t_1,\ldots,s,\eps)=\sum_{g\ge 0}\eps^{g-1}\mcF^o_g(t_0,t_1,\ldots,s):=\sum_{\substack{g,k,l\ge 0\\2g-2+k+2l>0}}\frac{\eps^{g-1}}{k!l!}\<\tau_{d_1}\cdots\tau_{d_l}\sigma^k\>^o_g t_{d_1}\cdots t_{d_l} s^k
$$
was defined as a direct generalization of $\mcF^c$; the new formal variable $s$ tracks the number of boundary marked points. The definition of the all-genus open intersection numbers can be found in~\cite{Tes15}, which summarizes the construction of~\cite{ST1}. A combinatorial formula for the open intersection numbers in all genera was given in~\cite{Tes15} and then used in~\cite{BT17} to prove the open analogue of Witten's conjecture. By~\cite{Bur15}, the generating series $\mcF^o$ gives a solution of a certain extension of the KdV hierarchy and is related to the wave function of the KdV hierarchy (corresponding to the tau-function $\exp(\mcF^c)$) by an explicit formula~\cite{Bur16}. 

\medskip

It is natural to ask whether these two generalizations of Witten's conjecture can be combined, producing an open analogue of Witten's $r$-spin conjecture.  Toward this end, in~\cite{BCT1} we defined a moduli space of {\it graded $r$-spin disks} $\M_{0,k,(a_1, \ldots, a_l)}^{1/r}$ as well as an open Witten bundle $\mathcal{W}$ and cotangent line bundles $\mathbb{L}_1, \ldots, \mathbb{L}_{l}$ at the internal marked points, and then in~\cite{BCT2} we defined the corresponding {\it open $r$-spin intersection numbers}\footnote{The construction from~\cite{PST14} is recovered as a special case, when $r=2$ and all $a_i$ are zero.}
\begin{equation*}
\<\tau_{d_1}^{a_1}\cdots\tau^{a_l}_{d_l}\sigma^k\>^{\frac{1}{r},o}_0,\quad 0\le a_i\le r-1,\quad d_i\ge 0.
\end{equation*}
Equipped with these numbers, we defined an \emph{open $r$-spin potential} in genus zero by
\[
\mcF^{\frac{1}{r},o}_0(t^*_*,s):=\sum_{\substack{k,l\geq 0\\k+2l>2}} \sum_{\substack{0 \leq a_1, \ldots, a_l \leq r-1\\ d_1, \ldots, d_l \geq 0}} \frac{1}{k!l!}\<\tau^{a_1}_{d_1}\cdots\tau^{a_l}_{d_l}\sigma^k\>^{\frac{1}{r},o}_0 t^{a_1}_{d_1} \cdots t^{a_l}_{d_l}s^k.
\]
We then considered a special solution 
$$
\phi=\sum_{g\ge 0}\eps^{g-1}\phi_g(t^*_*),\quad\phi_g\in \mbC[[t^*_*]]
$$
of a certain extension of the Gelfand--Dickey hierarchy and proved a formula for the generating series $\mcF^{\frac{1}{r},o}_0$ in terms of the formal power series $\phi_0$. We will recall the details in Section~\ref{section:GD hierarchy}.

\medskip

While these constructions and results are limited to genus zero, in~\cite{BCT2} we conjectured that for any genus $g\ge 1$ there is a geometric construction of open $r$-spin intersection numbers $\<\tau^{a_1}_{d_1}\cdots\tau^{a_l}_{d_l}\sigma^k\>^{\frac{1}{r},o}_g$ generalizing our construction in genus zero.  Given such intersection numbers, we defined a generating series~$\mcF_g^{\frac{1}{r},o}(t^*_*,s)$ by
$$
\mcF_g^{\frac{1}{r},o}(t^*_*,s):=\sum_{l,k\ge 0}\frac{1}{l!k!}\sum_{\substack{0\le a_1,\ldots,a_l\le r-1\\d_1,\ldots,d_l\ge 0}}\<\tau^{a_1}_{d_1}\cdots\tau^{a_l}_{d_l}\sigma^k\>^{\frac{1}{r},o}_g t^{a_1}_{d_1}\cdots t^{a_l}_{d_l}s^k,
$$
and conjectured an explicit formula for it in terms of the formal power series $\phi_g$.  See Conjecture~\ref{main conjecture} below for the explicit statement.

\medskip

In this paper, we study the main conjecture from~\cite{BCT2} in more detail. One can expect that the numbers $\<\tau^{a_1}_{d_1}\cdots\tau^{a_l}_{d_l}\sigma^k\>^{\frac{1}{r},o}_g$ satisfy a series of natural properties:
\begin{enumerate}
\item In the case $r=2$ and $a_1=\cdots=a_l=0$, it is natural to identify the corresponding open $r$-spin intersection numbers with the intersection numbers on $\oM_{g,k,l}$.

\smallskip

\item From the dimension of the moduli space of genus-$g$ $r$-spin surfaces with boundary and an expected formula for the degree of an open analogue of Witten's class, one can see that the open $r$-spin intersection numbers should vanish unless a dimension constraint is satisfied.

\smallskip

\item From a natural expectation for the behavior of an open analogue of Witten's class under the map forgetting a marked point of twist $0$, one can see that the generating series~$\mcF^{\frac{1}{r},o}_g$ should satisfy open string and open dilaton equations. The genus-zero part of these equations was proved in~\cite{BCT2}, while analogous equations for the intersection numbers~$\<\tau_{d_1}\cdots\tau_{d_l}\sigma^k\>^o_g$ were conjectured in~\cite{PST14} and proved in~\cite{BT17} (a geometric proof will appear in~\cite{ST1}).

\smallskip

\item Based on the genus-one topological recursion relations for the intersection numbers $\<\tau_{d_1}\ldots\tau_{d_l}\sigma^k\>^o_1$ conjectured by the authors of~\cite{PST14} and proved by Solomon and the third author (a geometric proof will appear in~\cite{ST1}), we expect a natural generalization of these relations to be satisfied by the open $r$-spin intersection numbers in genus one.  
\end{enumerate}

\medskip

The main result of this paper, proved in Section~\ref{section:evidence}, is that all of these expected properties of the open $r$-spin intersection numbers agree with Conjecture~\ref{main conjecture}.

\medskip

\noindent{\bf Convention.} We use the standard convention of sum over repeated Greek indices.

\medskip

\noindent{\bf Acknowledgements.} The work of A.~B. is an output of a research project implemented as part of the Basic Research Program at the National Research University Higher School of Economics (HSE University).  E.C. was supported by NSF CAREER grant 2137060.  R.T. (incumbent of the Lillian and George Lyttle Career Development Chair) was supported by a research grant from the Center for New Scientists of Weizmann Institute and by the ISF (grant No. 335/19).

\medskip

%%%%%%%%%%%%%%%%%%%%%%%%%%%%%%%%%%%%%%%%%%%%%%%%%%%%%%%%%%%%%%%%%%%%%%%%%%%%%%%%
%%%%%%%%%%%%%%%%%%%%%%%%%%%%%%%%%%%%%%%%%%%%%%%%%%%%%%%%%%%%%%%%%%%%%%%%%%%%%%%%

\section{The Gelfand--Dickey hierarchy and the main conjecture}\label{section:GD hierarchy}

Consider formal variables $T_i$ for $i\ge 1$. A \emph{pseudo-differential operator}~$A$ is a Laurent series
$$
A=\sum_{n=-\infty}^m a_n\d_x^n,\quad a_n\in\mbC[\eps,\eps^{-1}][[T_1,T_2,\ldots]],
$$
where $m$ is an integer and $\d_x$ is a formal variable. Denote
$$
A_+:=\sum_{n=0}^m a_n\d_x^n,\qquad A_-:=A-A_+,\qquad \res A:=a_{-1}.
$$
The space of such operators is endowed with the structure of a noncommutative associative algebra, in which the multiplication, denoted by~$\circ$, is defined by the formula
\begin{gather*}
\d_x^k\circ f:=\sum_{l=0}^\infty\frac{k(k-1)\cdots(k-l+1)}{l!}\frac{\d^l f}{\d T_1^l}\d_x^{k-l},\quad f\in\mbC[\eps,\eps^{-1}][[T_*]],\quad k\in\mbZ.
\end{gather*}
We identify $x=T_1$, and in the case $A_-=0$ we interpret $A$ as a differential operator acting in the space of formal power series $\mbC[[T_*]]$ in the obvious way.

\medskip

For any $r\ge 2$ and any pseudo-differential operator~$A$ of the form 
$$
A=\d_x^r+\sum_{n=1}^\infty a_n\d_x^{r-n},
$$ 
there exists a unique pseudo-differential operator $A^{\frac{1}{r}}$ of the form 
$$
A^{\frac{1}{r}}=\d_x+\sum_{n=0}^\infty \widetilde{a}_n\d_x^{-n}
$$
such that $\left(A^{\frac{1}{r}}\right)^r=A$.

\medskip

Let $r\ge 2$, and consider the pseudo-differential operator
$$
L:=\d_x^r+\sum_{i=0}^{r-2}f_i\d_x^i,\quad f_i\in\mbC[\eps,\eps^{-1}][[T_*]].
$$
For any $n\ge 1$, the commutator $[(L^{n/r})_+,L]$ has the form $\sum_{i=0}^{r-2}h_i\d_x^i$ with $h_i\in\mbC[\eps,\eps^{-1}][[T_*]]$. The \emph{$r$-th Gelfand--Dickey hierarchy} is the following system of partial differential equations for the formal power series $f_0,f_1,\ldots,f_{r-2}$:
\begin{gather*}
\frac{\d L}{\d T_n}=\eps^{n-1}[(L^{n/r})_+,L],\quad n\ge 1.
\end{gather*}

\medskip

Consider the solution $L$ of the Gelfand--Dickey hierarchy specified by the initial condition
\begin{gather}\label{eq:initial condition for GD}
L|_{T_{\ge 2}=0}=\d_x^r+\eps^{-r}rx.
\end{gather}
The $r$-spin Witten conjecture states that $\frac{\d\mcF^{\frac{1}{r},c}}{\d t^{r-1}_d}=0$ for $d\ge 0$, and under the change of variables
\begin{gather*}
T_k=\frac{1}{(-r)^{\frac{3k}{2(r+1)}-\frac{1}{2}-d}k!_r}t^a_d,\quad 0\le a\le r-2,\quad d\ge 0,
\end{gather*}
where $k=a+1+rd$ and $k!_r:=\prod_{i=0}^d(a+1+ri)$, we have
$$
\res L^{n/r}=\eps^{1-n}\frac{\d^2\mcF^{\frac{1}{r},c}}{\d T_1\d T_n},\quad n\ge 1,\quad r\nmid n.
$$

\medskip

With $L$ as above, let~$\Phi(T_*,\eps)\in\mbC[\eps,\eps^{-1}][[T_*]]$ be the solution of the system of equations
\begin{gather}\label{eq:system for Phi}
\frac{\d\Phi}{\d T_n}=\eps^{n-1}(L^{n/r})_+\Phi,\quad n\ge 1,
\end{gather}
that satisfies the initial condition $\left.\Phi\right|_{T_{\ge 2}=0}=1$. Consider the expansion
$$
\phi:=\log\Phi=\sum_{g\in\mbZ}\eps^{g-1}\phi_g,\quad \phi_g\in\mbC[[T_*]].
$$
Note that by~\cite[Lemma~4.4]{BCT_Closed_Extended} $\phi_g=0$ for $g<0$. Comparing to the formal power series~$\mcF^{\frac{1}{r},c}_0$, which depends only on the variables $t^0_d,\ldots,t^{r-2}_d$, the formal power series~$\mcF^{\frac{1}{r},o}_0$ depends also on~$t^{r-1}_d$ and~$s$. We relate the variables $T_{mr}$ and $t^{r-1}_{m-1}$ as follows:
\begin{gather*}
T_{mr}=\frac{1}{(-r)^{\frac{m(r-2)}{2(r+1)}}m!r^m}t^{r-1}_{m-1},\quad m\ge 1.
\end{gather*}

\medskip

In~\cite{BCT2} we proved the following result:
\begin{gather*}\label{eq:main result}
\mcF^{\frac{1}{r},o}_0=\frac{1}{\sqrt{-r}}\phi_0\big|_{t^{r-1}_d\mapsto \frac{1}{\sqrt{-r}}(t^{r-1}_d-r\delta_{d,0}s)}-\frac{1}{\sqrt{-r}}\phi_0\big|_{t^{r-1}_d\mapsto\frac{1}{\sqrt{-r}}t^{r-1}_d}.
\end{gather*}

\medskip

Regarding the open $r$-spin intersection numbers in higher genera, we proposed the following conjecture.

\begin{conj}[\cite{BCT2}]\label{main conjecture}
For any $g\ge 1$ we have
$$
\mcF^{\frac{1}{r},o}_g=\left.(-r)^{\frac{g-1}{2}}\phi_g\right|_{t^{r-1}_d\mapsto\frac{1}{\sqrt{-r}}(t^{r-1}_d-\delta_{d,0}rs)}.
$$
\end{conj}

\medskip

%%%%%%%%%%%%%%%%%%%%%%%%%%%%%%%%%%%%%%%%%%%%%%%%%%%%%%%%%%%%%%%%%%%%%%%%%%%%%%%%
%%%%%%%%%%%%%%%%%%%%%%%%%%%%%%%%%%%%%%%%%%%%%%%%%%%%%%%%%%%%%%%%%%%%%%%%%%%%%%%%

\section{Evidence for the main conjecture}\label{section:evidence}

\subsection{The case $r=2$}

Suppose $r=2$. In~\cite{BCT2} we proved that
\begin{gather}
\<\tau^0_{d_1}\cdots\tau^0_{d_l}\sigma^k\>^{\frac{1}{2},o}_0=(-2)^{\frac{k-1}{2}}\<\tau_{d_1}\cdots\tau_{d_l}\sigma^k\>^{o}_0,
\end{gather}
where we recall that the right-hand side refers to the intersection numbers on $\oM_{0,k,l}$. Regarding higher genus, in~\cite{BT17} the authors proved that
\begin{gather}\label{eq:open and phi}
\mcF^o_{g}(t_0,t_1,\ldots,s)=\phi_g|_{\substack{t^0_d=t_d\\t^1_d=\delta_{d,0}s}}.
\end{gather}
Thus, it is natural to identify
\begin{gather}\label{eq:r2 identification}
\<\tau^0_{d_1}\cdots\tau^0_{d_l}\sigma^k\>^{\frac{1}{2},o}_g:=(-2)^{\frac{g+k-1}{2}}\<\tau_{d_1}\cdots\tau_{d_l}\sigma^k\>^{o}_g.
\end{gather}
Indeed, note that the factor $2^{-\frac{g+k-1}{2}}$ is forcibly included in the definition of the intersection numbers~$\<\tau_{d_1}\cdots\tau_{d_l}\sigma^k\>^o_g$. Also, note that the intersection number $\<\tau_{d_1}\cdots\tau_{d_l}\sigma^k\>^o_g$ is zero unless $g+k$ is odd. Thus, changing the orientation of~$\oM_{g,k,l}$ by $(-1)^{\frac{g+k-1}{2}}$ when $g+k$ is odd gives the additional factor $(-1)^{\frac{g+k-1}{2}}$. This confirms Conjecture~\ref{main conjecture} in the case $r=2$ after setting $t^1_d=0$.

\medskip

\subsection{Dimension constraint}

It is natural to expect that the open $r$-spin intersection number
\begin{gather*}
\<\tau^{\alpha_1}_{d_1}\cdots\tau^{\alpha_l}_{d_l}\sigma^k\>^{\frac{1}{r},o}_g
\end{gather*}
is zero unless
\[
\frac{(g+k-1)(r-2)+2\sum\alpha_i}{r}+2\sum d_i=3g-3+2l+k.
\]
Indeed, the right-hand side is the dimension of the moduli space of $r$-spin genus-$g$ surfaces with boundary, while the left-hand side is the virtual dimension of the purported open analogue of the virtual fundamental cycle that should correspond to the intersection problem.  This is equivalent to the constraint
\begin{equation}\label{eq:dimension condition, all genera}
\sum\left(\frac{\alpha_i}{r}+d_i-1\right)-\frac{k}{r}=\frac{(r+1)(g-1)}{r}.
\end{equation}
On the other hand, in~\cite[Lemma~4.4]{BCT_Closed_Extended} we proved that the derivative $\left.\frac{\d^n\phi_g}{\d t^{\alpha_1}_{d_1}\cdots\d t^{\alpha_n}_{d_n}}\right|_{t^*_*=0}$ is zero unless $\sum\left(\frac{\alpha_i}{r}+d_i-1\right)=\frac{(r+1)(g-1)}{r}$. We see that if Conjecture~\ref{main conjecture} is true then this gives exactly the expected constraint for the open $r$-spin intersection numbers.

\medskip

\subsection{Open string and open dilaton equations}

One would expect the formal power series~$\mcF_g^{\frac{1}{r},o}$ to satisfy the open string and the open dilaton equations
\begin{align}
&\frac{\d\mcF_g^{\frac{1}{r},o}}{\d t^0_0}=\sum_{n\ge 0}t^\alpha_{n+1}\frac{\d\mcF_g^{\frac{1}{r},o}}{\d t^\alpha_n}+\delta_{g,0}s,\label{eq:open rspin string}\\
&\frac{\d\mcF_g^{\frac{1}{r},o}}{\d t^0_1}=(g-1)\mcF_g^{\frac{1}{r},o}+\sum_{n\ge 0}t^\alpha_n\frac{\d\mcF_g^{\frac{1}{r},o}}{\d t^\alpha_n}+s\frac{\d\mcF_g^{\frac{1}{r},o}}{\d s}+\delta_{g,1}\frac{1}{2}.\label{eq:open rspin dilaton}
\end{align}

\medskip

Indeed, in the case $r=2$ and $t^1_d=0$, using the identification~\eqref{eq:r2 identification}, these equations become exactly the open string and the open dilaton equations for the generating series $\mcF^o_g$, which where conjectured in~\cite{PST14}, proved there in genus $0$ using geometric technique, and proved in all genera in~\cite{BT17} using a matrix model (a geometric proof will appear in~\cite{ST1}).

\medskip

In~\cite[Proposition~5.2]{BCT2}, we proved equations~\eqref{eq:open rspin string} and~\eqref{eq:open rspin dilaton} in genus zero for any $r\ge 2$. Although this was done using an open-closed correspondence and open string and open dilaton equations for the closed extended $r$-spin intersection numbers, we could also prove them geometrically, imitating the proofs in~\cite{PST14}. In particular, we expect the geometric proof to work in all genera and all $r\geq 2$, given a construction of an open virtual fundamental cycle for the higher-genus Witten bundle that satisfies certain expected properties (being pulled back from the moduli without markings of twist zero, for example, and boundary behavior similar to that of canonical sections in genus zero). 

\medskip

By~\cite[Theorem~1.2]{BY15} (see also~\cite[Lemmas~4.5]{BCT_Closed_Extended}) we have
\begin{gather*}
\frac{\d\phi_g}{\d t^0_0}=\sum_{n\ge 0}t^\alpha_{n+1}\frac{\d\phi_g}{\d t^\alpha_n}+\delta_{g,0}t^{r-1}_0.
\end{gather*}
If Conjecture~\ref{main conjecture} is true, then this implies the expected equation~\eqref{eq:open rspin string}.  The expected equation \eqref{eq:open rspin dilaton}, on the other hand, is implied by Conjecture~\ref{main conjecture} by way of the following proposition.

\medskip

\begin{prop}
We have
\begin{gather}\label{eq:dilaton for phi}
\frac{\d\phi}{\d t^0_1}=\sum_{n\ge 0}t^\alpha_n\frac{\d\phi}{\d t^\alpha_n}+\eps\frac{\d\phi}{\d\eps}+\frac{1}{2}.
\end{gather}
\end{prop}
\begin{proof}
In the variables $T_i$, equation~\eqref{eq:dilaton for phi} looks as follows:
\begin{gather}\label{eq:dilaton for phi,2}
\left(\frac{1}{r+1}\frac{\d}{\d T_{r+1}}-\eps\frac{\d}{\d\eps}-\sum_{i\ge 1}T_i\frac{\d}{\d T_i}\right)\phi=\frac{1}{2}.
\end{gather}
Before we proceed with the proof, let us recall more facts from the theory of the Gelfand--Dickey hierarchy (see, e.g.,~\cite{Dic03}).

\medskip

Consider a pseudo-differential operator $A=\sum_{n=-\infty}^m a_n(T_*,\eps)\d_x^n$. The Laurent series
$$
\widehat A(T_*,\eps,z):=\sum_{n=-\infty}^m a_n(T_*,\eps)z^n,
$$
in which $z$ is a formal variable, is called the {\it symbol} of the operator~$A$. Suppose an operator~$L$ is a solution of the Gelfand--Dickey hierarchy. Then there exists a pseudo-differential operator~$P$ of the form
\begin{gather*}
P=1+\sum_{n\ge 1}p_n(T_*,\eps)\d_x^{-n},
\end{gather*}
satisfying $L=P\circ\d_x^r\circ P^{-1}$ and
\begin{gather*}
\frac{\d P}{\d T_n}=-\eps^{n-1}\left(L^{n/r}\right)_-\circ P,\quad n\ge 1.
\end{gather*}
The operator $P$ is called a \emph{dressing operator} of the operator $L$.

\medskip

Denote by $G_z$ the shift operator that acts on a formal power series $f\in\mbC[\eps,\eps^{-1}][[T_1,T_2,\ldots]]$ as follows:
\begin{gather*}
G_z(f)(T_1,T_2,T_3,\ldots):=f\left(T_1-\frac{1}{z},T_2-\frac{1}{2\eps z^2},T_3-\frac{1}{3\eps^2 z^3},\ldots\right).
\end{gather*}
Let $P=1+\sum_{n\ge 1}p_n(T_*,\eps)\d_x^{-n}$ be a dressing operator of some operator $L$ satisfying the Gelfand--Dickey hierarchy. Then there exists a series $\tau\in\mbC[\eps,\eps^{-1}][[T_1,T_2,T_3,\ldots]]$ with constant term $\left.\tau\right|_{T_i=0}=1$ for which
$$
\widehat P=\frac{G_z(\tau)}{\tau}.
$$
The series $\tau$ is called a \emph{tau-function} of the Gelfand--Dickey hierarchy. The operator $L$ can be reconstructed from the tau-function $\tau$ by the following formula:
$$
\res L^{n/r}=\eps^{1-n}\frac{\d^2\log\tau}{\d T_1\d T_n},\quad n\ge 1.
$$

\medskip

Denote the linear differential operator in the brackets on the left-hand side of equation~\eqref{eq:dilaton for phi,2} by~$O$. Let $L$ be the solution of the Gelfand--Dickey hierarchy specified by the initial condition~\eqref{eq:initial condition for GD}. Let us show that
\begin{gather}\label{eq:dilaton for L}
\left(z\frac{\d}{\d z}+O\right)\widehat L=r\widehat L.
\end{gather}
Witten's $r$-spin conjecture~\cite{Witten93}, proved by Faber--Shadrin--Zvonkine~\cite{FSZ10}, says that the formal power series~$\tau^{\frac{1}{r},c}:=\exp(\mcF^{\frac{1}{r},c})$ is a tau-function of the Gelfand--Dickey hierarchy corresponding to the operator~$L$. Therefore, a dressing operator $P$ of the operator $L$ is given by 
$$
\widehat P=\frac{G_z(\tau^{\frac{1}{r},c})}{\tau^{\frac{1}{r},c}}.
$$
The function~$\tau^{\frac{1}{r},c}$ satisfies the dilaton equation
$$
O\tau^{\frac{1}{r},c}=\frac{r-1}{24}\tau^{\frac{1}{r},c}.
$$
We compute
$$
\left(z\frac{\d}{\d z}+O\right)G_z(\tau^{\frac{1}{r},c})=G_z(O\tau^{\frac{1}{r},c})=\frac{r-1}{24}G_z(\tau^{\frac{1}{r},c}),
$$
and, thus,
$$
\left(z\frac{\d}{\d z}+O\right)\widehat P=\left(z\frac{\d}{\d z}+O\right)\frac{G_z(\tau^{\frac{1}{r},c})}{\tau^{\frac{1}{r},c}}=0.
$$
Note that the commutation relation $\left[O,\frac{\d}{\d T_1}\right]=\frac{\d}{\d T_1}$ implies that if $\left(z\frac{\d}{\d z}+O\right)\widehat{A}=a\widehat{A}$ and $\left(z\frac{\d}{\d z}+O\right)\widehat{B}=b\widehat{B}$ for some pseudo-differential operators $A,B$ and $a,b\in\mbZ$, then $\left(z\frac{\d}{\d z}+O\right)\widehat{A\circ B}=(a+b)\widehat{A\circ B}$. Since $L=P\circ\d_x^r\circ P^{-1}$, we conclude that equation~\eqref{eq:dilaton for L} is true.

\medskip

Note also that we have
\begin{gather}\label{eq:gendilaton for L}
\left(z\frac{\d}{\d z}+O\right)\widehat{L^{n/r}}=n\widehat{L^{n/r}},\quad n\ge 1,
\end{gather}
and that for any pseudo-differential operator $A$ and $f\in\mbC[[T_*]]$ we have
\begin{gather}\label{eq:identity for O}
O(A_+ f)=A_+(O f)+\left(\left.\left(z\frac{\d}{\d z}+O\right)\widehat{A}_+\right|_{z\mapsto\d_x}\right)f.
\end{gather}

\medskip

Let us finally prove equation~\eqref{eq:dilaton for phi,2}. Equivalently, we have to prove that $O\Phi=\frac{1}{2}\Phi$. We compute
\begin{align*}
&\left.L^{\frac{1}{r}}\right|_{T_{\ge 2}=0}=\d_x+\eps^{-r}x\d_x^{-r+1}-\frac{r-1}{2}\eps^{-r}\d_x^{-r}+\ldots,\\
&\left.(L^{\frac{r+1}{r}})_+\right|_{T_{\ge 2}=0}=\d_x^{r+1}+(r+1)\eps^{-r}x\d_x+\frac{r+1}{2}\eps^{-r}.
\end{align*}
Taking into account that $\Phi|_{T_{\ge 2}=0}=1$, we obtain
\begin{gather}\label{eq:initial condition for OPhi}
\left.O\Phi\right|_{T_{\ge 2}=0}=\left.\frac{1}{r+1}\frac{\d\Phi}{\d T_{r+1}}\right|_{T_{\ge 2}=0}=\left.\frac{\eps^r}{r+1}\widehat{(L^{\frac{r+1}{r}})_+}\right|_{z=T_{\ge 2}=0}=\frac{1}{2}.
\end{gather}
We also have
$$
\frac{\d}{\d T_n}(O\Phi)=(O-1)(\eps^{n-1}(L^{n/r})_+\Phi)\stackrel{\text{eqs.~\eqref{eq:gendilaton for L},\eqref{eq:identity for O}}}{=}\eps^{n-1}(L^{n/r})_+(O\Phi).
$$
We see that the formal power series $O\Phi$ satisfies the same system of PDEs~\eqref{eq:system for Phi} as the formal power series $\Phi$, with the initial condition~\eqref{eq:initial condition for OPhi}. Thus, $O\Phi=\frac{1}{2}\Phi$ and the proposition is proved.
\end{proof}

\medskip

Clearly, if Conjecture~\ref{main conjecture} is true, then the proposition implies the expected equation~\eqref{eq:open rspin dilaton}. 

\medskip

\subsection{Open topological recursion relations in genus one}

\begin{theorem}\label{theorem:TRR1}
We have
\begin{gather*}
\frac{\d\phi_1}{\d t^\alpha_{p+1}}=\sum_{\mu+\nu=r-2}\frac{\d^2\mcF^{\frac{1}{r},c}_0}{\d t^\alpha_p\d t^\mu_0}\frac{\d\phi_1}{\d t^\nu_0}+\frac{\d\phi_0}{\d t^\alpha_p}\frac{\d\phi_1}{\d t^{r-1}_0}+\frac{1}{2}\frac{\d^2\phi_0}{\d t^\alpha_p\d t^{r-1}_0},\quad 0\le\alpha\le r-1,\quad p\ge 0.
\end{gather*}
\end{theorem}

\medskip

Before proving the theorem, let us discuss some consequences. 

\begin{cor}
The generating series of intersection numbers on $\oM_{1,k,l}$ satisfies the relation
\begin{gather}\label{eq:open TRR-1}
\frac{\d\mcF^o_1}{\d t_{p+1}}=\frac{\d^2\mcF^c_0}{\d t_p\d t_0}\frac{\d\mcF^o_1}{\d t_0}+\frac{\d\mcF^o_0}{\d t_p}\frac{\d\mcF^o_1}{\d s}+\frac{1}{2}\frac{\d^2\mcF^o_0}{\d t_p\d s},\quad p\ge 0.
\end{gather}
\end{cor}
\begin{proof}
This follows from the theorem and equation~\eqref{eq:open and phi}.
\end{proof}

\medskip

The system of relations~\eqref{eq:open TRR-1} was conjectured by the authors of \cite{PST14} and called the \emph{open topological recursion relations in genus one}. A geometric proof will appear in~\cite{ST1}. 

\medskip

We expect the following topological recursion relations for the genus-one open $r$-spin intersection numbers for any $0\le\alpha\le r-1$ and $p\ge 0$:
\begin{gather}\label{eq:open genus 1 TRR}
\frac{\d\mcF^{\frac{1}{r},o}_1}{\d t^\alpha_{p+1}}=\sum_{\mu+\nu=r-2}\frac{\d^2\mcF^{\frac{1}{r},c}_0}{\d t^\alpha_p\d t^\mu_0}\frac{\d\mcF^{\frac{1}{r},o}_1}{\d t^\nu_0}+\frac{\d \mcF^{\frac{1}{r},\text{ext}}_0}{\d t^\alpha_p}\frac{\d\mcF^{\frac{1}{r},o}_1}{\d t^{r-1}_0}+\frac{\d\mcF^{\frac{1}{r},o}_0}{\d t^\alpha_p}\frac{\d\mcF^{\frac{1}{r},o}_1}{\d s}+\frac{1}{2}\frac{\d^2\mcF^{\frac{1}{r},o}_0}{\d t^\alpha_p\d s},
\end{gather}
 where $\mcF^{\frac{1}{r},\text{ext}}_0(t^0_*,\ldots,t^{r-1}_*)$ is the generating series of genus-zero closed extended $r$-spin intersection numbers defined in the same way as the usual genus-zero $r$-spin intersection numbers, but where exactly one of $a_i$-s is equal to~$-1$. In~\cite[Theorem~4.6]{BCT_Closed_Extended}, we proved that 
$$
\mcF^{\frac{1}{r},\text{ext}}_0=\left.\sqrt{-r}\phi_0\right|_{t^{r-1}_d\mapsto\frac{1}{\sqrt{-r}}t^{r-1}_d}.
$$
The geometric proof of~\eqref{eq:open TRR-1} from~\cite{ST1} cannot work for the new system~\eqref{eq:open genus 1 TRR}, due to the lack of a rigorous construction of the open virtual fundamental cycle. Still, as in the case of the open string and dilaton equations, the claim is expected to be true, and with a similar proof, under some mild assumptions on the open virtual fundamental cycle. 

\medskip

It is easy to see that if Conjecture~\ref{main conjecture} is true, then Theorem~\ref{theorem:TRR1} implies the system of relations~\eqref{eq:open genus 1 TRR}.

\medskip

\begin{proof}[Proof of Theorem~\ref{theorem:TRR1}]
Equivalently, we have to prove that
\begin{gather}\label{eq:TRR for phi1 in T-variables}
\frac{\d\phi_1}{\d T_{a+r}}=\sum_{b=1}^{r-1}\frac{a+r}{b(r-b)}\frac{\d^2\mcF^{\frac{1}{r},c}_0}{\d T_a\d T_b}\frac{\d\phi_1}{\d T_{r-b}}+\frac{a+r}{r}\frac{\d\phi_0}{\d T_a}\frac{\d\phi_1}{\d T_r}+\frac{a+r}{2r}\frac{\d^2\phi_0}{\d T_a\d T_r},\quad a\ge 1.
\end{gather}

\medskip

In~\cite[Lemma~4.2]{BCT_Closed_Extended}, we proved that the operator $L$ has the form 
$$
L=\d_x^r+\sum_{i=0}^{r-2}\sum_{j\ge 0}\eps^{i-r+j}f_i^{[j]}\d_x^i,\quad f_i^{[j]}\in\mbC[[T_*]].
$$
Denote 
$$
L_0:=\d_x^r+\sum_{i=0}^{r-2}f_i^{[0]}\d_x^i,\qquad L_1:=\sum_{i=0}^{r-2}f_i^{[1]}\d_x^i,\qquad (f_i^{[0]})^{(k)}:=\d_x^k f_i^{[0]}.
$$
For any $a\ge 1$, the Laurent series $\widehat{L_0^{\frac{a}{r}}}$ has the form 
$$
\widehat{L_0^{\frac{a}{r}}}=\sum_{i=-\infty}^a P_i((f^{[0]}_*)^{(*)})z^i,
$$
where $P_i$ are polynomials in $(f^{[0]}_j)^{(k)}$ for $0\le j\le r-2$ and $k\ge 0$. Let us assign to $(f^{[0]}_j)^{(k)}$ differential degree~$k$. Then we can decompose the polynomials~$P_i((f^{[0]}_*)^{(*)})$ as $P_i((f^{[0]}_*)^{(*)})=\sum_{m\ge 0} P_{i,m}((f^{[0]}_*)^{(*)})$, where a polynomial $P_{i,m}$ has differential degree~$m$. Introduce the following notations:
\begin{gather*}
\left(\widehat{L_0^{\frac{a}{r}}}\right)_m:=\sum_{i=-\infty}^a P_{i,m}((f^{[0]}_*)^{(*)})z^i,\qquad\left(\widehat{L_0^{\frac{a}{r}}}\right)_{+,m}:=\sum_{i=0}^a P_{i,m}((f^{[0]}_*)^{(*)})z^i.
\end{gather*}

\medskip

\begin{lemma}
For any $a\ge 1$ we have
\begin{gather}\label{eq:open GD in genus 1}
\frac{\d\phi_1}{\d T_a}=\left.\left[(\phi_1)_x\d_z\left(\hL_0^{\frac{a}{r}}\right)_++\frac{(\phi_0)_{xx}}{2}\d^2_z\left(\hL_0^{\frac{a}{r}}\right)_++\left(\widehat{L_0^{\frac{a}{r}}}\right)_{+,1}+\frac{a}{r}\left(\hL_0^{\frac{a}{r}-1}\hL_1\right)_+\right]\right|_{z=(\phi_0)_x},
\end{gather}
where $\d_z$ denotes the partial derivative $\frac{\d}{\d z}$.
\end{lemma}
\begin{proof}
We have $\frac{\d\phi}{\d T_a}=\eps^{a-1}\frac{(L^{a/r})_+e^\phi}{e^\phi}$. The operator $L^{\frac{a}{r}}$ has the form $L^{\frac{a}{r}}=\sum_{i=-\infty}^aR_i\d_x^i$, where $R_i=\sum_{j\ge 0}\eps^{i-a+j}R_{i,j}$, $R_{i,j}\in\mbC[[T_*]]$. We have to check that
\begin{multline}\label{eq:open GD,genus 1,tmp1}
\Coef_{\eps^0}\left[\eps^{a-1}\sum_{i=0}^aR_i\frac{\d_x^ie^\phi}{e^\phi}\right]=\\
=\left.\left[(\phi_1)_x\d_z\left(\hL_0^{\frac{a}{r}}\right)_++\frac{(\phi_0)_{xx}}{2}\d^2_z\left(\hL_0^{\frac{a}{r}}\right)_++\left(\widehat{L_0^{\frac{a}{r}}}\right)_{+,1}+\frac{a}{r}\left(\hL_0^{\frac{a}{r}-1}\hL_1\right)_+\right]\right|_{z=(\phi_0)_x}.
\end{multline}

\medskip

By the induction on $i$, it is easy to prove that
\begin{gather*}
\frac{\d_x^i e^\phi}{e^\phi}=i!\sum_{\substack{m_1,m_2,\ldots\ge 0\\\sum jm_j=i}}\prod_{j\ge 1}\frac{(\d_x^j\phi)^{m_j}}{(j!)^{m_j}m_j!},\quad i\ge 0.
\end{gather*}
We have
$$
\prod_{j\ge 1}(\d_x^j\phi)^{m_j}=
\begin{cases}
\eps^{-i}(\phi_0)_x^i+\eps^{-i+1}i(\phi_0)_x^{i-1}(\phi_1)_x+O(\eps^{-i+2}),&\text{if $m_1=i$ and $m_{\ge 2}=0$},\\
\eps^{-i+1}(\phi_0)_x^{i-2}(\phi_0)_{xx}+O(\eps^{-i+2}),&\text{if $m_1=i-2$, $m_2=1$ and $m_{\ge 3}=0$},\\
O(\eps^{-i+2}),&\text{otherwise}.
\end{cases}
$$
As a result,
$$
\frac{\d_x^i e^\phi}{e^\phi}=\eps^{-i}(\phi_0)_x^i+\eps^{-i+1}\left(i(\phi_0)_x^{i-1}(\phi_1)_x+\frac{i(i-1)}{2}(\phi_0)_x^{i-2}(\phi_0)_{xx}\right)+O(\eps^{-i+2}),
$$
and
\begin{gather}\label{eq:open GD,genus 1,tmp2}
\Coef_{\eps^0}\left[\eps^{a-1}\sum_{i=0}^aR_i\frac{\d_x^ie^\phi}{e^\phi}\right]=\sum_{i=0}^a R_{i,0}\left(i(\phi_0)_x^{i-1}(\phi_1)_x+\frac{i(i-1)}{2}(\phi_0)_x^{i-2}(\phi_0)_{xx}\right)+\sum_{i=0}^aR_{i,1}(\phi_0)_x^i.
\end{gather}
Note that
$$
\sum_{i=0}^aR_{i,0}z^i=\left(\hL_0^{\frac{a}{r}}\right)_+,\qquad \sum_{i=0}^a R_{i,1}z^i=\left(\widehat{L_0^{\frac{a}{r}}}\right)_{+,1}+\frac{a}{r}\left(\hL_0^{\frac{a}{r}-1}\hL_1\right)_+.
$$
We can see now that the first sum on the right-hand side of~\eqref{eq:open GD,genus 1,tmp2} gives the first and the second terms on the right-hand side of~\eqref{eq:open GD,genus 1,tmp1}. The second sum on the right-hand side of~\eqref{eq:open GD,genus 1,tmp2} gives the third and the fourth terms on the right-hand side of~\eqref{eq:open GD,genus 1,tmp1}. This completes the proof of the lemma.
\end{proof}

\medskip

In~\cite[Lemma~4.7]{BCT_Closed_Extended}, we proved that
\begin{gather}\label{eq:open GD in genus 0}
\frac{\d\phi_0}{\d T_a}=\left.\left(\hL_0^{\frac{a}{r}}\right)_+\right|_{z=(\phi_0)_x},
\end{gather}
which implies
$$
\frac{\d^2\phi_0}{\d T_a\d T_r}=\frac{\d}{\d T_r}\left[\left.\left(\hL^{\frac{a}{r}}_0\right)_+\right|_{z=(\phi_0)_x}\right]=\left.\left((\phi_0)_{xx}\d_z\left(\hL_0^\frac{a}{r}\right)_+\d_z\hL_0+\d_z\left(\hL_0^\frac{a}{r}\right)_+\d_x\hL_0\right)\right|_{z=(\phi_0)_x}.
$$
Therefore, equation~\eqref{eq:TRR for phi1 in T-variables} is equivalent to
\begin{align*}
\frac{\d\phi_1}{\d T_{a+r}}=&\sum_{b=1}^{r-1}\frac{a+r}{b(r-b)}\frac{\d^2\mcF^{\frac{1}{r},c}_0}{\d T_a\d T_b}\frac{\d\phi_1}{\d T_{r-b}}+\frac{a+r}{r}\frac{\d\phi_0}{\d T_a}\frac{\d\phi_1}{\d T_r}+\\
&+\frac{a+r}{2r}\left.\left((\phi_0)_{xx}\d_z\left(\hL_0^\frac{a}{r}\right)_+\d_z\hL_0+\d_z\left(\hL_0^\frac{a}{r}\right)_+\d_x\hL_0\right)\right|_{z=(\phi_0)_x}.
\end{align*}
By formulas~\eqref{eq:open GD in genus 1} and~\eqref{eq:open GD in genus 0}, this is equivalent to the equation
\begin{align*}
&\underline{(\phi_1)_x\d_z\left(\hL_0^{\frac{a+r}{r}}\right)_+}+\underline{\underline{\frac{(\phi_0)_{xx}}{2}\d^2_z\left(\hL_0^{\frac{a+r}{r}}\right)_+}}+\boxed{\left(\widehat{L_0^{\frac{a+r}{r}}}\right)_{+,1}}+\frac{a+r}{r}\left(\hL_0^{\frac{a}{r}}\hL_1\right)_+=\\
=&\sum_{b=1}^{r-1}\frac{a+r}{b(r-b)}\frac{\d^2\mcF^{\frac{1}{r},c}_0}{\d T_a\d T_b}\left[\underline{(\phi_1)_x\d_z\left(\hL_0^{\frac{r-b}{r}}\right)_+}+\underline{\underline{\frac{(\phi_0)_{xx}}{2}\d^2_z\left(\hL_0^{\frac{r-b}{r}}\right)_+}}+\boxed{\left(\widehat{L_0^{\frac{r-b}{r}}}\right)_{+,1}}+\frac{r-b}{r}\left(\hL_0^{-\frac{b}{r}}\hL_1\right)_+\right]\\
&+\frac{a+r}{r}\hL_0^{\frac{a}{r}}\left[\underline{(\phi_1)_x\d_z\hL_0}+\underline{\underline{\frac{(\phi_0)_{xx}}{2}\d^2_z\hL_0}}+\hL_1\right]\\
&+\frac{a+r}{2r}\left(\underline{\underline{(\phi_0)_{xx}\d_z\left(\hL_0^\frac{a}{r}\right)_+\d_z\hL_0}}+\boxed{\d_z\left(\hL_0^\frac{a}{r}\right)_+\d_x\hL_0}\right)
\end{align*}
where we should substitute $z=(\phi_0)_x$. Collecting together the terms marked in the same way, we see that this equation is a consequence of the following four equations:
\begin{align}
&\d_z\left(\hL_0^\frac{a+r}{r}\right)_+=\sum_{b=1}^{r-1}\frac{a+r}{b(r-b)}\frac{\d^2\mcF^{\frac{1}{r},c}_0}{\d T_a\d T_b}\d_z\left(\hL_0^\frac{r-b}{r}\right)_++\frac{a+r}{r}\left(\hL_0^{\frac{a}{r}}\right)_+\d_z\hL_0,\label{eq:open TRR-1,first equation}\\
&\d^2_z\left(\hL_0^\frac{a+r}{r}\right)_+=\sum_{b=1}^{r-1}\frac{a+r}{b(r-b)}\frac{\d^2\mcF^{\frac{1}{r},c}_0}{\d T_a\d T_b}\d^2_z\left(\hL_0^\frac{r-b}{r}\right)_++\frac{a+r}{r}\left(\hL_0^\frac{a}{r}\right)_+\d^2_z\hL_0+\label{eq:open TRR-1,second equation}\\
&\hspace{2.5cm}+\frac{a+r}{r}\d_z\left(\hL_0^\frac{a}{r}\right)_+\d_z\hL_0,\notag\\
&\left(\widehat{L^\frac{a+r}{r}_0}\right)_{+,1}=\sum_{b=1}^{r-1}\frac{a+r}{b(r-b)}\frac{\d^2\mcF^{\frac{1}{r},c}_0}{\d T_a\d T_b}\left(\widehat{L^\frac{r-b}{r}_0}\right)_{+,1}+\frac{a+r}{2r}\d_z\left(\hL^\frac{a}{r}_0\right)_+\d_x\hL_0,\label{eq:open TRR-1,third equation}\\
&\frac{a+r}{r}\left(\hL_0^{\frac{a}{r}}\hL_1\right)_+=\sum_{b=1}^{r-1}\frac{a+r}{b(r-b)}\frac{\d^2\mcF^{\frac{1}{r},c}_0}{\d T_a\d T_b}\frac{r-b}{r}\left(\hL_0^{-\frac{b}{r}}\hL_1\right)_++\frac{a+r}{r}\left(\hL_0^{\frac{a}{r}}\right)_+\hL_1.\label{eq:open TRR-1,fourth equation}
\end{align}

\medskip

In~\cite[equation~(4.19)]{BCT_Closed_Extended}, we proved that
\begin{gather*}
d\left(\hL_0^\frac{a+r}{r}\right)_+=\sum_{b=1}^{r-1}\frac{a+r}{b(r-b)}\frac{\d^2\mcF^{\frac{1}{r},c}_0}{\d T_a\d T_b}d\left(\hL_0^\frac{r-b}{r}\right)_++\frac{a+r}{r}\left(\hL_0^{\frac{a}{r}}\right)_+d\hL_0,
\end{gather*}
from which equation~\eqref{eq:open TRR-1,first equation} clearly follows.  Applying the derivative $\d_z$ to both sides of~\eqref{eq:open TRR-1,first equation}, we get equation~\eqref{eq:open TRR-1,second equation}. Equation~\eqref{eq:open TRR-1,fourth equation} follows from the property~\cite[equation~(4.22)]{BCT_Closed_Extended}
\begin{gather}\label{eq:second important identity}
\left(\hL_0^{\frac{a}{r}}\right)_--\sum_{b=1}^{r-1}\frac{1}{b}\frac{\d^2\mcF^{\frac{1}{r},c}_0}{\d T_a\d T_b}\hL_0^{-\frac{b}{r}}\in z^{-r-1}\mbC[f^{[0]}_*][[z^{-1}]].
\end{gather}

\medskip

It remains to prove equation~\eqref{eq:open TRR-1,third equation}. Let us first prove that for any $a\ge 1$ we have
\begin{gather}\label{eq:degree 1 part of a power of L}
\left(\widehat{L^{\frac{a}{r}}_0}\right)_{1}=\frac{a(a-r)}{2r^2}\hL_0^{\frac{a}{r}-2}\d_z\hL_0\d_x\hL_0.
\end{gather}
It is easy to see that
\begin{align*}
\left(\widehat{\left(L_0^\frac{1}{r}\right)^r}\right)_{1}=&r\hL_0^\frac{r-1}{r}\left(\widehat{L_0^{\frac{1}{r}}}\right)_{1}+\sum_{j=0}^{r-1}\hL_0^\frac{r-j-1}{r}\d_z\hL_0^\frac{1}{r}\d_x\hL_0^\frac{j}{r}=\\
=&r\hL_0^\frac{r-1}{r}\left(\widehat{L_0^{\frac{1}{r}}}\right)_{1}+\sum_{j=0}^{r-1}\frac{j}{r^2}\hL_0^{-1}\d_z\hL_0\d_x\hL_0=\\
=&r\hL_0^\frac{r-1}{r}\left(\widehat{L_0^{\frac{1}{r}}}\right)_{1}+\frac{r-1}{2r}\hL_0^{-1}\d_z\hL_0\d_x\hL_0.
\end{align*}
Since we obviously have $\left(\widehat{\left(L^\frac{1}{r}\right)^r}\right)_{1}=0$, equation~\eqref{eq:degree 1 part of a power of L} is proved for $a=1$. For an arbitrary $a\ge 1$, we compute
\begin{align*}
\left(\widehat{\left(L_0^\frac{1}{r}\right)^a}\right)_{1}=&a\hL_0^\frac{a-1}{r}\left(\widehat{L_0^{\frac{1}{r}}}\right)_{1}+\sum_{j=0}^{a-1}\hL_0^\frac{a-1-j}{r}\d_z\hL_0^\frac{1}{r}\d_x\hL_0^\frac{j}{r}=\\
=&\frac{a(1-r)}{2r^2}\hL_0^{\frac{a}{r}-2}\d_z\hL_0\d_x\hL_0+\sum_{j=0}^{a-1}\frac{j}{r^2}\hL_0^{\frac{a}{r}-2}\d_z\hL_0\d_x\hL_0=\\
=&\frac{a(a-r)}{2r^2}\hL_0^{\frac{a}{r}-2}\d_z\hL_0\d_x\hL_0.
\end{align*}
Thus, equation~\eqref{eq:degree 1 part of a power of L} is proved.

\medskip

Let us rewrite formula~\eqref{eq:degree 1 part of a power of L} in the following way: $\left(\widehat{L_0^{\frac{a}{r}}}\right)_{1}=\frac{a}{2r}\d_z\hL_0^{\frac{a}{r}-1}\d_x\hL_0$. Then we see that equation~\eqref{eq:open TRR-1,third equation} is equivalent to
\begin{gather}\label{eq:proof of equation 3 for open TRR-1}
\frac{a+r}{2r}\left(\d_z\hL_0^\frac{a}{r}\d_x\hL_0\right)_+=\sum_{b=1}^{r-1}\frac{a+r}{2rb}\frac{\d^2\mcF^{\frac{1}{r},c}_0}{\d T_a\d T_b}\left(\d_z\hL_0^{-\frac{b}{r}}\d_x\hL_0\right)_++\underline{\frac{a+r}{2r}\d_z\left(\hL_0^\frac{a}{r}\right)_+\d_x\hL_0}.
\end{gather}
We obviously have
$$
\frac{a+r}{2r}\left(\d_z\hL_0^\frac{a}{r}\d_x\hL_0\right)_+=\underline{\frac{a+r}{2r}\d_z\left(\hL_0^\frac{a}{r}\right)_+\d_x\hL_0}+\frac{a+r}{2r}\left(\d_z\left(\hL_0^\frac{a}{r}\right)_-\d_x\hL_0\right)_+.
$$
Note that the underlined term here cancels the underlined term on the right-hand side of~\eqref{eq:proof of equation 3 for open TRR-1}. Therefore, equation~\eqref{eq:proof of equation 3 for open TRR-1} is equivalent to the identity
\begin{gather*}
\left(\d_z\left(\hL_0^\frac{a}{r}\right)_-\d_x\hL_0\right)_+=\sum_{b=1}^{r-1}\frac{1}{b}\frac{\d^2\mcF^{\frac{1}{r},c}_0}{\d T_a\d T_b}\left(\d_z\hL_0^{-\frac{b}{r}}\d_x\hL_0\right)_+,
\end{gather*}
which follows from equation~\eqref{eq:second important identity}. The theorem is proved.
\end{proof}

\medskip

\bibliographystyle{abbrv}
\bibliography{OpenBiblio}

\end{document}